\documentclass{article}
\usepackage{spconf}
\usepackage{graphicx}
\usepackage{bm}
\usepackage{booktabs}
\usepackage{subfig}
\usepackage[font=footnotesize]{caption}
\usepackage{hyperref}
\usepackage[vlined]{algorithm2e}
\usepackage{nicefrac}
\usepackage{siunitx}

\usepackage{widebar}
\usepackage{localmath}

\providecommand{\half}{\nicefrac{1}{2}}
\newcommand{\Sph}{\mathbb{S}}

\usepackage{textcomp}
\usepackage[absolute,showboxes]{textpos}

\TPMargin{5pt}

\newcommand{\copyrightstatement}{
	\begin{textblock}{0.84}(0.08,0.01)    
		\noindent
		\footnotesize
		\textcopyright~2021~IEEE.  Personal  use  of  this  material  is  permitted.
		Permission from IEEE must be obtained for all other uses, in any current or
		future media, including reprinting/republishing this material for advertising
		or promotional purposes, creating new collective works, for resale or
		redistribution to servers or lists, or reuse of any copyrighted component of
		this work in other works.
	\end{textblock}
}

\graphicspath{{figures/}}

\begin{document}
\ninept

\copyrightstatement

\title{Refinement of Direction of Arrival Estimators by Majorization-Minimization Optimization on the Array Manifold}

\name{Robin Scheibler and Masahito Togami%
\thanks{
Code available at \protect\url{https://github.com/fakufaku/doamm}.%
}}
\address{LINE Corporation, Tokyo, Japan}

%

\maketitle

\begin{abstract}
  We propose a generalized formulation of direction of arrival estimation that includes many existing methods such as steered response power, subspace, coherent and incoherent, as well as speech sparsity-based methods.
  Unlike most conventional methods that rely exclusively on grid search, we introduce a continuous optimization algorithm to refine DOA estimates beyond the resolution of the initial grid.
  The algorithm is derived from the majorization-minimization (MM) technique.
  We derive two surrogate functions, one quadratic and one linear.
  Both lead to efficient iterative algorithms that do not require hyperparameters, such as step size, and ensure that the DOA estimates never leave the array manifold, without the need for a projection step.
  In numerical experiments, we show that the accuracy after a few iterations of the MM algorithm nearly removes dependency on the resolution of the initial grid used.
  We find that the quadratic surrogate function leads to very fast convergence, but the simplicity of the linear algorithm is very attractive, and the performance gap small.
\end{abstract}


\section{Introduction}

Direction of arrival (DOA) estimation is a technique to spatially localize sources using an array of sensors~\cite{Krim:1996cq}.
The focus of this paper is on acoustic sound source localization, where sensors are microphones~\cite{Brandstein:2010uh}.
However, we note that DOA estimation is also widely applied in radar~\cite{Vasanelli:2020gc} and sonar~\cite{Grall:2020ik} technologies.

There has been a tremendous amount of work on DOA over the years.
Steered response power (SRP) methods create beams in multiple directions to find those with largest power~\cite{Capon:1969da,DiBiase:2000uv,Tashev:2009vj}.
Subspace methods, such as the famous MUSIC algorithm, identify source directions due to their orthogonality to the noise subspace~\cite{Schmidt:1986js}.
An important development is that of sparsity-based methods, starting with~\cite{Malioutov:2005jw}.
All the methods so far operate on a grid and their accuracy is limited by its coarseness.
Grid-search is costly in terms of processing power and memory to store all the steering vectors.
These problems become acute for full three dimensional localization where accuracy within a degrees requires over 30000 grid points.
The requirements are then multiplied by the number of frequency sub-bands and sensors used.
As a mitigation, so-called off-the-grid methods combine sparse optimization with a continuous re-gridding step~\cite{Gretsistas:2012tj,Zhang:2019jb}.
For some special array geometries, e.g. linear, the MUSIC criterion can be solved by polynomial rooting~\cite{Barabell:1983dm,Bresler:1986hqa}.
Extensions to arbitrary arrays and the wideband setting exists via array interpolation~\cite{Friedlander:1993ew} and finite rate of innovation~\cite{Pan:2017bm,Pan:2019dv}.
Recently, several deep learning based systems have been proposed~\cite{Adavanne:2019bz,Chakrabarty:2019hm}.

While more complex methods lead to improved accuracy and performance, classic methods such as SRP and MUSIC are reliable, easy to implement, and widely deployed in the industry.
In this work, we introduce a simple algorithm to refine coarse grid estimates produced by these methods using continuous domain optimization based on the majorization-minimization (MM) technique~\cite{Lange:2016wp,Sun:2017hb}.
Recently, the algorithm SPIRE-MM~\cite{Togami:2020tp} has been proposed for robust speech DOA estimation.
It introduced MM optimization refinement of time-frequency bin-wise DOA estimators with a surrogate function based on a cosine inequality~\cite{Yamaoka:2019wb}.
We propose a generalization of this MM algorithm applicable to most of the SRP and subspace based methods.
We formulate the problem as the minimization of the \textit{power mean}~\cite{Xu:2019ek} of the cost at the different frequency bands.
We produce two surrogate functions, one quadratic, similar to~\cite{Togami:2020tp}, and one linear.
The former can be minimized efficiently by the generalized trust region sub-problem method~\cite{MORE:1993dy}, as in~\cite{Togami:2020tp}, and the latter has a closed-form solution.
Both lead to iterative algorithms that are guaranteed to decrease the value of the objective.
They estimate DOA vectors directly on the sphere, without any projection, and do not require hyperparameters.
One iteration has the same computational cost as the evaluation of one grid point.

We study the performance of the refinement method for SRP-PHAT~\cite{DiBiase:2000uv} and MUSIC~\cite{Schmidt:1986js} on both simulated and recorded datasets.
We find that the median accuracy with initial estimates computed on a 100-points grid is better than that of a 10000-points grid.
Furthermore, the number of iteration required is around 10 for the quadratic surrogate, and 30 for the linear one.

\section{Background}
\seclabel{background}

We consider an array of $M$ sensors with known locations $\vd_m\in\R^3$, $m=1,\ldots,M$.
We work in the time-frequency domain by applying a short-time Fourier transform~\cite{Allen:1977in} to the input time domain signals.
The signal model of the vector containing the sensor measurements at frequency band $k=1,\ldots,K$ and frame $n=1,\ldots,N$, is 
\begin{equation}
  \vx_{kn} = \sum\nolimits_{\ell=1}^L \va_k(\vq_\ell) y_{\ell kn} + \vb_{kn},
  \elabel{signal_model}
\end{equation}
with $\vx_{kn} \in \C^M$, the $\ell$th source signal $y_{\ell kn} \in \C$, and the noise $\vb_{kn} \in \C^M$.
The vector $\vq_\ell\in \R^3$ is a unit-length vector pointing towards source $\ell$.
In terms of colatitude $\theta$ and azimuth $\varphi$, we have $\vq = \begin{bmatrix} \cos \varphi \sin\theta & \sin\varphi \sin\theta & \cos\theta \end{bmatrix}^\top$.
The steering vectors are 
\begin{align}
  \va_k(\vq) & = M^{-\half} \Big[\, e^{j\omega_k \vd_1^\top \vq} \quad \cdots \quad e^{j\omega_k \vd_M^\top \vq}\, \Big]^\top.
\end{align}
with $\omega_k = 2\pi \frac{f_k}{c}$ where $f_k$ is the center frequency of the $k$th band and $c$ the speed of sound.

In the rest of this manuscript, we denote vectors and matrices by bold lower and upper case letters, respectively.
Furthermore, $\mA^\top$ and $\mA^\H$ denote the transpose and conjugate transpose, respectively of matrix $\mA$.
The Euclidean norm of vector $\vx$ is written $\|\vx\| = (\vx^\H\vx)^{\half}$.
Unless specified otherwise, indices $m$, $k$, $n$, and $\ell$ take the ranges defined here.
We use $\R_+$ to denote the set of non-negative real numbers.
The $(d-1)$-sphere is $\Sph^{d-1} = \{\vu \,|\, \vu \in \R^d, \| \vu\| = 1\}$.
The DOA vectors belong to the 2-sphere, i.e. $\vq \in \Sph^2$.
Finally, we define the unnormalized sinc function as $\sinc(x) = \sin(x) / x$.

\subsection{Covariance-based DOA Estimators}
\seclabel{cov_based_doa}

Many of the conventional DOA estimators can be described as finding the local maxima, one per source, of the function,
\begin{align}
  \calJ(\vq) = \sum\nolimits_{k=1}^K (\va_k(\vq)^\H \mC_k \va_k(\vq))^s,\quad \vq \in \Sph^2,
  \elabel{general_doa_estimator}
\end{align}
where $\mC_k\in\C^{M\times M}$ is a Hermitian positive semi-definite matrix related to the covariance of the input channels and $s$ is a real exponent.
If we consider $K=1$, this is a narrowband estimator.
When working with the sum of the objectives of narrow-bands, i.e. $K>1$, these methods are known as \textit{incoherent}.
See \cite{Krim:1996cq,Tashev:2009vj} for more details.

\textbf{The Steered Response Power} family of algorithms is obtained with $\mC_k = \mS_k$, a weighted sample covariance matrix of \eref{signal_model},
\begin{align}
  \mS_k = N^{-1} \sum\nolimits_{n=1}^N \tilde{\vx}_{kn} \tilde{\vx}_{kn}^\H,
  \elabel{input_covariance}
\end{align}
and $s=1$.
The vector $\tilde{\vx}_{kn}$ is a weighted version of $\vx_{kn}$, i.e. $\tilde{x}_{mkn} = w_{mkn} x_{mkn}$.
Choosing $w_{mkn} = 1$ or $w_{mkn} = |x_{mkn}|^{-1}$  yield the conventional SRP or SRP-PHAT estimators~\cite{DiBiase:2000uv}, respectively.
Other weighting schemes used for the generalized cross-correlation (GCC)~\cite{Knapp:1976el}, such as SCOT, are also possible.

\textbf{MUSIC} ~\cite{Schmidt:1986js,Barabell:1983dm,Friedlander:1993ew} considers the covariance matrix of \eref{signal_model}, i.e.,
\begin{equation}
  \mR_{k} = \Expect{\vx_{kn} \vx_{kn}^\H} = \sum\nolimits_\ell \Expect{|y_{\ell kn}|^2} \va_k(\vq_\ell)\va_k(\vq_\ell)^\H + \mB_{k},
  \nonumber
\end{equation}
where $\mB_k$ is the covariance matrix of the noise.
Now, let the eigenvalue decomposition of $\mS_k \approx \mR_k$ (with $w_{kn} = 1$) be $\mS_k = \mU_k \mLambda_k \mU_k^\H$, where $\mU_k$ contains the eigenvectors of $\mS_{k}$ and $\mLambda_k$ is the diagonal matrix containing the eigenvalues.
MUSIC decomposes $\mU_k$ into signal and noise subspaces, $\mU_k = [\mG_{k}\  \mE_{k}]$.
These subspaces are assumed orthogonal, i.e.  $\mE_{k}^\H \va_k(\vq_\ell) = 0$, for all $\ell$, and the estimator uses $\mC_k = \mE_{k} \mE_{k}^\H$ and $s=-1$.

\textbf{The MVDR} estimator, also known as \textit{Capon}, is obtained by taking $\mC_k = \mS_k^{-1}$ (with $w_{kn}=1$) and $s=-1$.

\textbf{Wideband Coherent Subspace Methods}, e.g.,~\cite{Wang:1985jv,diClaudio:2001bb}, construct an estimate of the covariance matrix at one frequency $f_{k^\prime}$ using the data from \textit{all} $K$ frequencies.
From this estimate, they obtain a basis $\mE_{k^\prime}$ for the noise subspace at $f_{k^\prime}$ and the DOA is obtained from \eref{general_doa_estimator} with $\mC_k = \mE_{k^\prime} \mE_{k^\prime}^\H$ and $s=-1$.

\textbf{Robust Speech DOA Methods}, e.g., \cite{Togami:2007vp,Togami:2007jm,Togami:2020tp}, are robust DOA estimators exploiting the sparseness of speech signals.
Namely, they assume so-called \textit{W-disjoint} orthogonality whereas each time-frequency bin is occupied by at most one source.
They share the following general structure.
First, for all $k, n$, compute the local DOA estimate $\vq_{kn}$ with \eref{general_doa_estimator} with $s=1$ and $\mC_k = \tilde{\vx}_{kn} \tilde{\vx}_{kn}^\H$, where $\tilde{\vx} = \vx_{kn} / |\vx_{kn}|$.
Then, compute the histogram of $\vq_{kn}$ on a finer grid.
Pick the peaks of the histogram as source locations.

\subsection{MM Algorithm for unit-min-sum-cos}
\seclabel{unit_min_sum_cos}

\begin{definition}[unit-min-sum-cos]
  The unit-min-sum-cos problem is the minimization of a sum of cosine subject to a unit norm constraint,
  \begin{equation}
    \min_{\vq\in \Sph^{d-1}} \sum_{p\in\calP} u_p \cos(\psi_p - \vb_p^\top \vq),
    \elabel{unit-min-sum-cos}
  \end{equation}
  where $u_p \in\R_+$, $\psi_p\in\R$, $\vb_p \in \R^d$, and $\calP$ some index set.
\end{definition}
Because $\mC_k$ is positive semi-definite and $\va_k(\vq)$ are complex exponentials, the terms $\va_k(\vq)^\H \mC_k \va_k(\va)$ from \eref{general_doa_estimator} are of the form \eref{unit-min-sum-cos} (up to a constant).
No closed form solution to minimize or maximize \eref{unit-min-sum-cos} exists, but MM optimization is applicable~\cite{Togami:2020tp}.
A quadratic surrogate for~\eref{unit-min-sum-cos} can be obtained from the following proposition~\cite{Yamaoka:2019wb}.
\begin{proposition}[Cosine Surrogate~\cite{Yamaoka:2019wb}]
  \label{proposition:cosine_surrogate}
  Let $\theta,\theta_0\in\R$, $z_0 = \underset{z\in\Z}{\arg\min}\ \left| \theta_0 + 2\pi z\right|$, and $\phi_0 = \theta_0 + 2 \pi z_0$.
  Then, the following inequality holds, with equality when $\theta = \theta_0$,
  \begin{equation}
    -\cos \theta \leq  \half \sinc(\phi_0)(\theta + 2\pi z_0)^2 - \cos \phi_0 - \half\, \phi_0 \sin \phi_0.
  \end{equation}
\end{proposition}
The MM method leads to the following iterative updates for $\vq$
\begin{equation}
  \vq_{t} \gets \underset{\vq\in\Sph^2}{\arg\min}\ \sum_{p\in\calP} \hat{u}_{p}(\vq_{t-1}) (\hat{\psi}_{p}(\vq_{t-1}) - \vb_{p}^\top \vq)^2,
  \elabel{mm_quad_update}
\end{equation}
where the previous iterate $\vq_{t-1}$ is used to compute the quantities,
\begin{align}
  \hat{\psi}_p(\hat{\vq}) & \gets \psi_p + \pi  + 2\pi\, \underset{z\in\Z}{\arg\min}\left| \psi_p + \pi - \vb_p^\top \hat{\vq} + 2\pi z\right|, \elabel{umsc_psi} \\
  \hat{u}_p(\hat{\vq}) & \gets u_p \sinc(\hat{\psi}_p^{(t)} - \vb_p^\top \hat{\vq}).
  \elabel{umsc_u}
\end{align}
The update \eref{mm_quad_update} involves a quadratically constrained quadratic minimization.
This is known as a generalized trust region sub-problem, and, albeit non-convex, its global minimum can be computed efficiently as follows~\cite{MORE:1993dy}.
Rewrite~\eref{mm_quad_update} as a quadratic form,
\begin{align}
  \vq_t \gets \underset{\vq\in\Sph^d}{\arg\min}\ \vq^\top \mD \vq - 2\vv^\top \vq,\quad
  \mD\in\R^{d\times d}, \vv\in \R^d.
  \elabel{gtrs}
\end{align}
By the method of Lagrange multipliers, the solution is $\vq_t = \hat{\vq}(\mu_0)=(\mD + \mu_0 \mI)^{-1}\vv$, where $\mu_0$ is the unique zero of $\|\hat{\vq}(\mu)\|^2 - 1$ larger than $-\lambda_{\min}$, with $\lambda_{\min}$ the smallest eigenvalue of $\mD$.
The zero can be efficiently found by Newton-Raphson or bisection, and working in the eigenspace of $\mD$ (see~\cite{MORE:1993dy,Togami:2020tp} for details).

\section{MM Algorithms to Refine DOA Estimators}
\seclabel{algorithm}

We start by stating a general formulation covering all the algorithms of \sref{cov_based_doa}.
Then, we introduce two surrogate functions for the generalized objective function, one quadratic, and one linear.

\subsection{Generalized Cost Function}

Instead of the maxima of \eref{general_doa_estimator}, we propose to find the local minima of
\begin{align}
  \calG(\vq) = \left[ \frac{1}{K} \sum\nolimits_k \left(\va_k^\H(\vq) \mV_k \va_k(\vq)\right)^s \right]^{\nicefrac{1}{s}}, \quad \vq \in \Sph^2,
  \elabel{proposed_objective}
\end{align}
with $s \in (-\infty, 1)\setminus_{\{0\}}$.
The function $M_s(\vy) = \left[ \frac{1}{K}\sum_k y_k^s \right]^{\nicefrac{1}{s}}$ with $y_k \geq 0, \forall k$, is a generalized power-mean~\cite{Xu:2019ek}.
This choice is motivated in part because it allows to formulate all the methods of \sref{cov_based_doa} by the same objective, and in part to explore new ways of combining multiple frequency bands.
For example, when $s \to -\infty$, only the frequency sub-band with largest power is considered.
We can show that with a careful choice of $\mV_k$ and $s$, we can recover all the methods of \sref{cov_based_doa}.

Those algorithms from \sref{cov_based_doa} that aim at maximizing $\va_k^\H(\vq)\mC_k\va_k(\vq)$ can be turned to minimization problems by considering $\mV_k = P_k - \mC_k$, with $P_k$ a positive constant.
By the Gershgorin circle theorem, choosing $P_k$ as the largest row-sum of $\mC_k$ ensures that $\mV_k$ is positive semi-definite.
Methods from \sref{cov_based_doa} that minimize $\va_k^\H(\vq)\mC_k\va_k(\vq)$, e.g. MUSIC-like, are readily covered with $\mV_k = \mC_k$.

Starting from a local estimate obtained from a coarse grid search, we perform local optimization by MM with a surrogate function $Q_v(\vq, \hat{\vq})$.
This leads to the following iterative algorithm,
\begin{align}
\vq_{t} & \gets \underset{\vq\in\Sph^2}{\arg\min}\ Q_v(\vq, \hat{\vq}_{t-1}),\quad v=1,2,
  \elabel{mm_update}
\end{align}
for $t=1,\ldots,T$.
By virtue of the MM algorithm, these updates are guaranteed to monotonically decrease the value of the objective of \eref{proposed_objective}~\cite{Lange:2016wp}.
The procedure is summarized in~\algref{doamm}.

\subsection{Quadratic Surrogate}

A quadratic surrogate function is obtained by successive application of the tangent inequality and Proposition~\ref{proposition:cosine_surrogate}.
\begin{theorem}
  Let $\calP = \{ (m, r)\,|\, 1\leq m\leq M, m < r \leq M \}$ be the set of distinct pairs, and for $p=(m,r)$, $\vDelta_p = \vd_m - \vd_r$, vector of coordinate differences of sensor $m$ and $r$, and
  \begin{align}
    u_{kp} & = \left| (\mV_k)_{mr} \right|, \quad
    \psi_{kp} = \arg\left((\mV_k)_{mr}\right).
  \end{align}
  Then, the following is a surrogate function of the objective of~\eref{proposed_objective},
  \begin{align}
    Q_2(\vq, \hat{\vq}) & = \vq^\top \mD(\hat{\vq}) \vq - 2 \vv(\hat{\vq})^\top \vq + \text{const.},
    \elabel{quadratic_surrogate}
  \end{align}
  where $\mD(\hat{\vq}) \in \R^{3\times 3}$ and $\vv(\hat{\vq})\in \R^3$ are defined as
  \begin{align}
    \mD(\hat{\vq}) & = \sum_{p\in\calP} \xi_p(\hat{\vq}) \vDelta_p \vDelta_p^\top, \quad \vv(\hat{\vq}) = \sum_{p\in\calP} \gamma_p(\hat{\vq}) \vDelta_p,
  \end{align}
  with weights
  \begin{align}
    \xi_p(\hat{\vq}) & = \sum\nolimits_k \omega_k^2 \beta_k(\hat{\vq}) \hat{u}_{kp}(\hat{\vq}), \\
    \gamma_p(\hat{\vq}) & = \sum\nolimits_k \omega_k \beta_k(\hat{\vq}) \hat{u}_{kp}(\hat{\vq}) \hat{\psi}_{kp}(\hat{\vq}),
  \end{align}
  where $\hat{\psi}_{kp}(\hat{\vq})$ and $\hat{u}_{kp}(\hat{\vq})$ are given by \eref{umsc_psi} and \eref{umsc_u} (with $\vb_p$ replaced by $\omega_k \vDelta_p$), respectively, and
  \begin{equation}
    \beta_k(\hat{\vq}) = \frac{\frac{1}{K} (\va_k^\H(\hat{\vq}) \mV_k \va_k(\hat{\vq}))^{s - 1}}{\left(\frac{1}{K} \sum_{k^\prime} (\va_{k^\prime}^\H(\hat{\vq}) \mV_{k^\prime} \va_{k^\prime}(\hat{\vq}))^s \right)^{1 - \nicefrac{1}{s}}}.
    \elabel{beta}
  \end{equation}
\end{theorem}
\begin{proof}
  For $s \in (-\infty,  1)$, one can show that $M_s(\vy)$ is concave~\cite{Xu:2019ek} and thus, the tangent inequality may be applied, giving,
  \begin{align}
    M_s(\vy) & \leq \nabla M_s(\hat{\vy})^\top \vy + \text{const}.
  \end{align}
  Replacing $\hat{\vy}$ by the terms of~\eref{proposed_objective} evaluated at $\hat{\vq}$ gives $\beta_k(\hat{\vq})$ of \eref{beta}.

  Due to the hermitian symmetry of $\mV_k$, the imaginary part of the symmetric terms in the quadratic form cancels and it can be expressed as a sum of cosine.
  This sum of cosine is majorized as in \sref{unit_min_sum_cos},
  \begin{align}
    \va_k^\H(\vq) \mV_k \va_k(\vq) & = \frac{\trace(\mV_k)}{M} + \frac{2}{M}\sum_{p\in\calP} u_{kp} \cos(\psi_{kp} - \omega_k \vDelta_p^\top \vq) \nonumber \\
                                   & \leq \sum_{p\in\calP} \hat{u}_{kp}(\hat{\vq}) (\hat{\psi}_{kp}(\hat{\vq})- \omega_k \vDelta_p^\top \vq)^2 + \text{const.}, \nonumber
  \end{align}
  where $\hat{\psi}_{kp}$ and $\hat{u}_{kp}$ are given by \eref{umsc_psi} and \eref{umsc_u}, respectively.

  Finally, notice that the terms in $\vq$ only have a scalar dependency on $k$.
  Expanding the quadratic terms, inverting the order of the sums, and completing the squares yields the final weights.
\end{proof}
Our first algorithm is obtained by solving \eref{mm_update} with $Q_2(\vq, \hat{\vq})$.
The minimizer of $Q_2(\vq, \hat{\vq})$ is computed exactly as that of \eref{gtrs}.

\begin{figure*}
    \centering
    \includegraphics[width=0.95\linewidth]{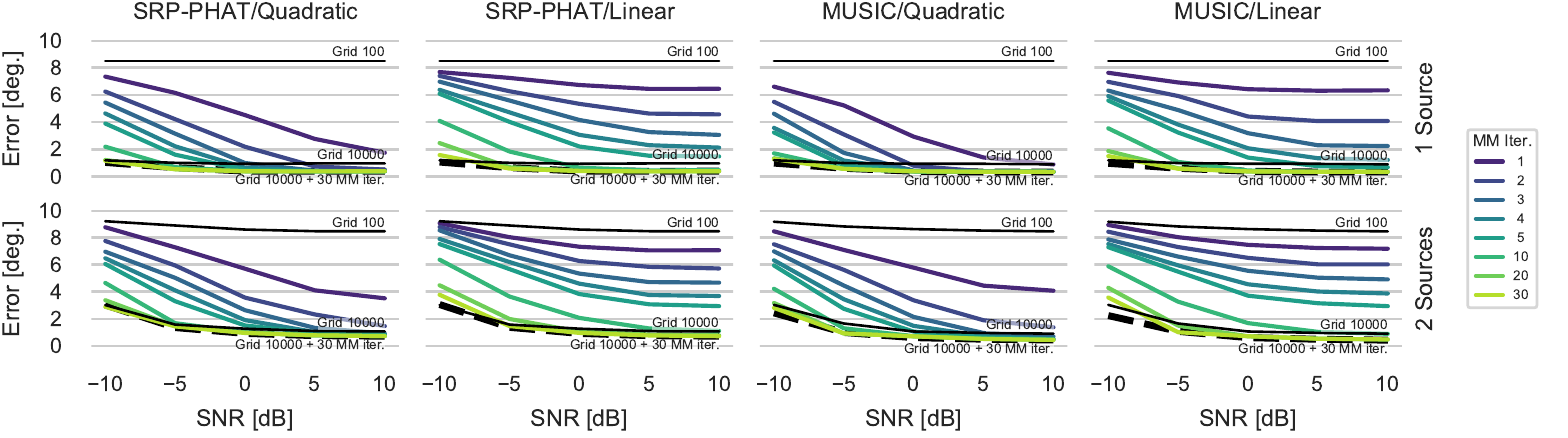}
    \caption{Median error in degrees as a function of the SNR. The thin black lines are for different grid sizes. The colored lines are for different number of iterations of the MM refinement algorithm starting from the 100-points grid estimate. The black thick dashed line is for 30 iterations of refinement started from the 10000-points grid estimate.}
    \flabel{grid_effect}
    \vspace{-0mm}
\end{figure*}

\subsection{Linear Surrogate}

Because a quadratic function with a norm constraint can be majorized with a linear function~\cite{Sun:2017hb}, we can derive an even simpler algorithm.
\begin{theorem}
  Let $\vq$ be such that $\|\vq\| = 1$, and
  \begin{align}
    C(\hat{\vq}) = \left(\max\nolimits_m \ \xi_p(\hat{\vq})\right) \; \lambda_{\max}\left(\sum\nolimits_p \vDelta_p\vDelta_p^\top \right)
    \elabel{lin_major_constant}
  \end{align}
  where $\lambda_{\max}(\,.\,)$ is the largest eigenvalue operator. Then,
  \begin{align}
    Q_1(\vq, \hat{\vq}) = 2 ((\mD(\hat{\vq}) - C(\hat{\vq})\mI) \hat{\vq} - \vv(\hat{\vq}))^\top\vq + \text{const.},
    \elabel{linear_surrogate}
  \end{align}
  is a surrogate function of~\eref{proposed_objective}.
\end{theorem}
\begin{proof}
  Dropping the dependency on $\hat{\vq}$ to simplify notation, we first show that $C\mI \succeq \mD$.
  For any $\vx$ such that $\|\vx\| = 1$, because $\xi_p\geq 0$,
  \begin{align}
    \vx^\top \left(\sum_p \xi_p \vDelta_p\vDelta_p^\top\right) \vx & \leq (\max\nolimits_p \ \xi_p) \; \lambda_{\max}\left(\sum\nolimits_p \vDelta_p\vDelta_p^\top\right). \nonumber
  \end{align}
  Then, we have the following majorization, with equality for $\vq=\hat{\vq}$,
  \begin{align}
    \vq^\top \mD \vq & \leq C \vq^\top \vq - 2 \hat{\vq}^\top (C\mI - \mD) \vq + \hat{\vq}^\top (C\mI - \mD)\hat{\vq}.
  \end{align}
  Due to the norm constraint $\vq^\top\vq = 1$, the quadratic term becomes constant and we obtain~\eref{linear_surrogate}.
\end{proof}
Now, the corresponding MM update is given by solving \eref{mm_update} for $v=1$.
The method of Lagrange multipliers yields,
\begin{align}
  \vq_t \gets \frac{\vv(\vq_{t-1}) - \mD(\vq_{t-1}) \vq_{t-1} + C(\vq_{t-1}) \vq_{t-1}}{\|\vv(\vq_{t-1}) - \mD(\vq_{t-1}) \vq_{t-1} + C(\vq_{t-1}) \vq_{t-1}\|},
  \elabel{linear_update}
\end{align}
which is indeed very simple.
This surrogate function has several advantages.
It does not require to solve any linear system.
In addition, the eigenvalue in~\eref{lin_major_constant} only depends on the array architecture and can be computed offline.
The price to pay is the slower convergence speed due to the extra majorization, as shown in \sref{experiments}.

\begin{algorithm}[t]
\SetKwInOut{Input}{Input}\SetKwInOut{Output}{Output}
\Input{$\mV_k$, $k=1,\ldots,K$}
\Output{$\vq_\ell$, $\ell=1,\ldots,L$}
\DontPrintSemicolon
Evaluate the cost function~\eref{proposed_objective} on a rough grid\;
Pick $L$ local minimizers $\vq_{\ell}$ for $\ell=1,\ldots,L$\;
\For{$\ell$ $\leftarrow 1$ \KwTo $L$}{
  \For{loop $\leftarrow 1$ \KwTo Max Iterations}{
    Update $\vq_\ell$ according to~\eref{gtrs} or~\eref{linear_update}
  }
}
\vspace{2mm}
\caption{Proposed DOA with MM refinement algorithm.}
\label{alg:doamm}
\end{algorithm}

\subsection{Computational Complexity}

The evaluation of $\hat{u}_{kp}, \hat{\psi}_{kp}, \beta_k, \xi_p, \gamma_p$ all require $K \times |\calP| = O(M^2 \times K)$ operations.
Thus, the overall complexity for both algorithms is $O(T M^2 K)$, $T$ being the number of MM iterations.
The linear surrogate provides some minor computational savings, but, as we will see, is slower to converge.
For comparison, the computation of the covariance matrices is $O(M^2 K N)$, subspace decomposition in MUSIC-like methods is $O(M^3 K)$, and the initial search on a grid of size $G$ is $O(G M^2 K)$.
Thus, the refinement algorithm represents a small fraction of the total computations.

\begin{table}
  \centering
  \caption{Median error in degrees for different values of $s$.}
  \resizebox{\linewidth}{!}{%
  \begin{tabular}{@{}llrrrrrrrr@{}}
    \toprule
    Srcs & \multicolumn{1}{r}{$s=$} &  -10.0 &  -3.0  &  -1.0  &  -0.5  &   0.2  &   0.5  &   0.8  &   1.0  \\
    \midrule
     1 & MUSIC &   0.40 &   0.40 &   0.40 &   0.40 &   0.40 &   0.40 &   0.40 &   0.40 \\
       & SRP-PHAT &   0.50 &   0.43 &   0.41 &   0.40 &   0.40 &   0.40 &   0.40 &   0.40 \\
     2 & MUSIC &   0.79 &   0.70 &   0.62 &   0.60 &   0.57 &   0.58 &   0.60 &   0.59 \\
       & SRP-PHAT &   0.79 &   \textbf{0.69} &   0.73 &   0.74 &   0.86 &   0.92 &   0.96 &   1.07 \\
     3 & MUSIC &   1.68 &   1.31 &   1.11 &   1.06 &   0.99 &   0.98 &   0.98 &   0.98 \\
       & SRP-PHAT &   1.43 &   \textbf{1.23} &   1.37 &   1.38 &   1.42 &   1.48 &   1.66 &   1.79 \\
       \bottomrule
  \end{tabular}
}
  \tlabel{value_s}
\end{table}

\begin{table}
  \centering
  \caption{Median runtimes in second with the quadratic surrogate.}
  \resizebox{0.9\linewidth}{!}{%
    \begin{tabular}{@{}lrrcccc@{}}
    \toprule
     & & & \multicolumn{2}{c}{SRP-PHAT} & \multicolumn{2}{c}{MUSIC} \\
     \cline{4-7}
      Description & Grid & MM Iter.  & 1 src & 2 src & 1 src & 2 src\\
    \midrule
      \textbf{proposed method} & 100 & 30 & 0.35 & 0.42 & 0.27 & 0.37 \\
      fine grid-search & 10000 & 0 & 4.55 & 4.58 & 4.57 & 4.48\\
    \bottomrule
  \end{tabular}
}
\tlabel{runtimes}
\end{table}

\section{Experiments}
\seclabel{experiments}

\subsection{Experiment on Simulated Reverberant Speech}

We use the \texttt{pyroomacoustics} toolbox~\cite{Scheibler:2018di} to simulate a cubic room of \SI{10}{\meter} side.
The reverberation time is set to \SI{500}{\milli\second}.
We use an array with twelve microphones whose locations are sub-sampled from the Pyramic array geometry~\cite{Scheibler:2018da} and placed at the center of the room.
The simulation is repeated for 100 random locations of the sources at \SI{3}{\meter} from the center of the array.
Uncorrelated additive white Gaussian noise added to the microphone inputs.
The DOAs are estimated using SRP-PHAT and MUSIC with different grid sizes and number of iterations of the refinement procedure.
The evaluation is done with respect to the permutation producing the smallest average error in terms of great circle distance.

\ffref{grid_effect} shows the median error as a function of the signal-to-noise ratio for the grid size of 100 and 10000, and the for successive refinement iterations starting from the former estimate.
We observe that 10 to 20 iterations is sufficient for the refinement procedure to produce better estimates regardless of the SNR or number of sources.
For higher SNR, the convergence is faster.
The linear surrogate requires slightly more iterations.
We also see that starting from the finer grid does not significantly improve the final estimate.
In \tref{value_s}, we compare the median error for different values of $s$.
For MUSIC, values of $s\geq-1$ seem appropriate in all conditions.
For SRP-PHAT, $s=-3$ produces the best results for two and three sources.
\tref{runtimes} shows between 11 to 17 times speed-up of the median runtime.

\begin{figure}
    \centering
    \includegraphics[width=0.8\linewidth]{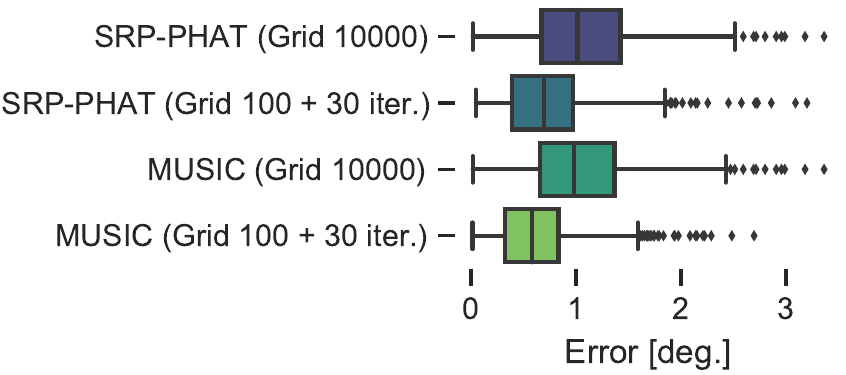}
    \caption{Error distribution on recorded data from the Pyramic dataset}
    \flabel{pyramic_experiment}
    \vspace{-0mm}
\end{figure}

\subsection{Experiment on the Pyramic Dataset}

We validate the proposed method on the Pyramic dataset of anechoic recordings with a 48-channel three-dimensional microphone array~\cite{Scheibler:2018da}.
We use one speech sample recorded with azimuth of \SI{0}{\degree} to \SI{358}{\degree} in \SI{2}{\degree} increments and colatitudes of \SI{77}{\degree}, \SI{89}{\degree}, and \SI{106}{\degree}.
Compare a using a 10000-points grid only, and a 100-points grid followed by 30 iterations of the refinement with the quadratic surrogate.
Box-plots of The DOA estimation error are shown in \ffref{pyramic_experiment}.
As in simulation, the rough grid followed by refinement consistently outperforms the fine grid.
Since one iteration of refinement has the same cost as a grid point evaluation, the proposed method, in this case, requires theoretically around $76\times$ less computations, ignoring the initial cost of computing $\mV_k$'s, which is the same in both cases.

\section{Conclusion}
\seclabel{conclusion}

We presented a refinement algorithm for classical DOA estimators such as SRP or MUSIC.
The algorithm relies on the MM technique to perform local improvements to the cost function starting from an initial estimate produced with a rough grid.
We proposed two different algorithms based on a quadratic and linear surrogate of the cost function, respectively.
The linear surrogate is especially attractive due to the simplicity of its implementation.
We find that the final median accuracy does not depend on the resolution of the initial grid.
This means that we can dramatically reduce the computational complexity by replacing fine grid-search by local optimization.
We demonstrated lower median error and 11--17 times practical runtime improvement for one to three sources on reverberant simulated data and an anechoic dataset of recordings.

\clearpage

\bibliographystyle{IEEEtran}
\bibliography{IEEEabrv,refs}

\end{document}